\pdfoutput=1
\documentclass[a4paper,10pt,twocolumn]{article}

\usepackage[a4paper,textwidth=17.0cm,textheight=23.6cm,columnsep=0.7cm,top=2.4cm,hmargin=2.0cm]{geometry}
\usepackage{amsmath,amssymb,amsthm}
\usepackage{graphicx}
\usepackage{times}
\usepackage[hidelinks]{hyperref}

\newtheorem{proposition}{Proposition}
\newtheorem{remark}{Remark}

\makeatletter
\def\@maketitle{%
  \begin{center}
    {\Large\bfseries Persistence of Bounded Rationality in Financial Markets:\\
     An Evolutionary Game-Theoretic Model\par}
    \vskip 1.2em
    {\large Cheng ZHANG\par}
    \vskip 0.8em
    {\normalsize Best AI Co., Ltd., Tokyo, Japan\par}
    \vskip 1.2em
    \parbox{15cm}{\small\textbf{Abstract:} \@abstract}
    \vskip 1.6em
  \end{center}}
\long\def\abstracttext#1{\long\gdef\@abstract{#1}}
\makeatother

\abstracttext{%
Classical theory predicts that irrational traders lose money to rational traders and are eventually driven
out of the market. Yet noise trading, herding, and heuristic trading persist even in highly developed markets.
We model a financial market as an ecology of two trading-behavior types, relatively rational (R) and boundedly
rational (N), whose composition evolves under replicator dynamics according to relative performance. Each
type's fitness depends on the share of the other type and on frictions. When the interaction is complementary,
that is, rational traders profit from the mispricing created by boundedly rational traders while the latter
benefit from the liquidity and price correction supplied by the former, and when the difference in frictions is
moderate, a unique interior composition exists and is globally asymptotically stable: bounded rationality
persists endogenously. When the interaction is antagonistic, the same point becomes a tipping threshold between
two one-type regimes. Comparative statics show that lowering the costs of disciplined trading, as algorithmic and
AI-based trading does, shifts the composition toward rational trading continuously but does not eliminate bounded
rationality unless the cost advantage exceeds an explicit threshold.}

\begin{document}
\maketitle
\thispagestyle{empty}

\section{Introduction}

The assumption of rational agents has long been the foundation of economic theory and financial models.
Market participants are assumed to form unbiased expectations, process information efficiently, and
maximize expected utility. Excessive price fluctuations, momentum and reversal, bubbles and crashes, and
persistent mispricing, however, suggest that many participants systematically deviate from this benchmark.

The classical defense of rationality is a selection argument: traders who act irrationally lose money to
rational traders and are eventually driven out of the market~\cite{friedman1953}. If this argument were
complete, bounded rationality would be a transient phenomenon. In reality, it is a persistent feature of
markets, including markets dominated by professional and algorithmic traders.

This paper examines the selection argument itself. We view a financial market as an aggregate of
heterogeneous trading-behavior types whose shares change according to relative performance, and ask under
what conditions selection eliminates bounded rationality and under what conditions it preserves it. In this
perspective, rationality is not an axiom but an outcome of selection, adaptation, and survival.

The contributions are as follows. (i) We formulate a two-type model in which relatively rational and
boundedly rational behaviors interact through cross-type effects and frictions, and evolve under replicator
dynamics. (ii) We prove that, when the interaction is complementary and the friction difference is moderate,
the market converges to a unique mixed composition from any interior initial state, so bounded rationality
persists endogenously. (iii) We show that when the interaction is antagonistic, the same point becomes a
tipping threshold and the market is bistable. (iv) We derive comparative statics that clarify how cost
reductions brought by algorithmic and AI-based trading shift the long-run composition.

\section{Related Work}

Bounded rationality goes back to Simon~\cite{simon1955}. Friedman's selection
argument~\cite{friedman1953} underpins the efficient market hypothesis~\cite{fama1970}. De Long
et al.~\cite{delong1990} showed that noise traders can generate price distortions that rational arbitrage
does not eliminate, and Shleifer and Vishny~\cite{shleifer1997} emphasized the limits of arbitrage. Kogan
et al.~\cite{kogan2006} and Blume and Easley~\cite{blume2006} analyzed when market selection does and does
not favor accurate beliefs.

Replicator dynamics were introduced by Taylor and Jonker~\cite{taylor1978}; their stability properties are
treated in~\cite{weibull1995,hofbauer1998,sandholm2010}. In finance, Kirman~\cite{kirman1993},
Lux~\cite{lux1995}, and Brock and Hommes~\cite{brock1998} showed how heterogeneous strategies and adaptive
switching generate rich market dynamics; Hommes~\cite{hommes2006} surveys this literature. Farmer~\cite{farmer2002}
described markets as an ecology of strategies, and Lo's adaptive markets hypothesis~\cite{lo2004} reconciles
efficiency with behavioral deviations through competition and natural selection.

This paper shares the ecological view of Farmer and Lo, but formulates it as an explicit replicator model in
which the sign of the cross-type interaction and the difference in frictions determine whether bounded
rationality persists.

\section{Model}

\subsection{Trader types}
Consider a market with many traders and two representative behavior types: (i) R-type (relatively
rational): traders with higher information-processing ability who use disciplined decision rules; and
(ii) N-type (boundedly rational): traders who rely on emotions, heuristics, or trend-following and
systematically deviate from rational benchmarks. Let $x\in[0,1]$ denote the share of R-type traders and
$1-x$ the share of N-type traders.

\subsection{Fitness with cross-type interaction}
We assume each type's fitness depends on the share of the \emph{other} type:
\begin{align}
\pi_R(x) &= a + b(1-x) - \kappa_R, \label{eq:piR}\\
\pi_N(x) &= c + d x - \kappa_N, \label{eq:piN}
\end{align}
where $a$ and $c$ are baseline fitness terms, $b$ and $d$ measure the cross-type interaction, and
$\kappa_R,\kappa_N$ capture frictions such as transaction costs, execution delays, regulatory constraints,
and technological or cognitive costs.

The signs of $b$ and $d$ determine the nature of the coupling. In the baseline case $b,d>0$, the interaction
is \emph{complementary}. The term $b(1-x)$ represents the profits R-type traders earn by trading against the
mispricing created by N-type traders, in the spirit of De Long et al.~\cite{delong1990}: the more noise, the
more opportunities for disciplined strategies. The term $dx$ represents the benefit N-type traders receive
from R-type traders, who supply liquidity and pull prices toward fundamentals, reducing the losses and
execution costs of heuristic trading. We also consider the \emph{antagonistic} case $b+d<0$, in which each
type mainly degrades the environment of the other.

\subsection{Replicator dynamics}
Let $n_R(t)$ and $n_N(t)$ denote the population masses (or capital) of the two behaviors and assume
proportional growth driven by fitness, $\dot n_R=n_R\pi_R$ and $\dot n_N=n_N\pi_N$. Differentiating
$x=n_R/(n_R+n_N)$ gives the replicator dynamics~\cite{taylor1978,weibull1995}
\begin{equation}
\dot{x}=x\bigl(\pi_R(x)-\bar{\pi}(x)\bigr)=x(1-x)\bigl[\pi_R(x)-\pi_N(x)\bigr],
\label{eq:rep}
\end{equation}
where $\bar\pi=x\pi_R+(1-x)\pi_N$. The dynamics do not assume that each trader optimizes; they describe
a selection process in which behaviors with relatively high performance expand through imitation and capital
flows. This differs from reinforcement learning, which describes the policy updates of a single agent.

\section{Analysis}
\label{sec:analysis}

Define the fitness difference $\Delta(x):=\pi_R(x)-\pi_N(x)$. From \eqref{eq:piR}--\eqref{eq:piN},
\begin{equation}
\Delta(x)=\theta-(b+d)x,\qquad \theta:=(a+b-c)-(\kappa_R-\kappa_N),
\label{eq:delta}
\end{equation}
where $\theta=\Delta(0)$ is the advantage of the R-type when it is rare. The boundary points $x=0,1$ are
always fixed points of \eqref{eq:rep}. Any interior fixed point must satisfy $\Delta(x)=0$, and since
$\Delta$ is affine, the only candidate is
\begin{equation}
x^*=\frac{\theta}{b+d}=\frac{(a+b-c)-(\kappa_R-\kappa_N)}{b+d}.
\label{eq:xstar}
\end{equation}

\subsection{Complementary interaction: stable coexistence}

\begin{proposition}[Existence and global stability]\label{prop:1}
Suppose $b+d>0$. (i) If
\begin{equation}
0<\theta<b+d,
\label{eq:range}
\end{equation}
then $x^*\in(0,1)$ is the unique interior fixed point, and every trajectory of \eqref{eq:rep} with
$x(0)\in(0,1)$ converges to $x^*$. (ii) If $\theta\le0$, then $x(t)\to0$; if $\theta\ge b+d$, then
$x(t)\to1$.
\end{proposition}

\begin{proof}
(i) Under \eqref{eq:range}, $x^*\in(0,1)$ and $\Delta(x)=(b+d)(x^*-x)$, so for $x\in(0,1)$ the sign of
$\dot x$ equals the sign of $x^*-x$. A trajectory starting in $(0,x^*)$ increases and one starting in
$(x^*,1)$ decreases; being monotone and bounded, each converges to a fixed point, which can only be $x^*$.
Local asymptotic stability also follows from
\begin{equation}
\left.\frac{d\dot x}{dx}\right|_{x^*}=x^*(1-x^*)\Delta'(x^*)=-(b+d)x^*(1-x^*)<0.
\label{eq:lin}
\end{equation}
(ii) If $\theta\le0$, then $\Delta(x)<0$ on $(0,1)$ and $x$ decreases to $0$. If $\theta\ge b+d$, then
$\Delta(x)\ge(b+d)(1-x)>0$ on $(0,1)$ and $x$ increases to $1$.
\end{proof}

Condition \eqref{eq:range} has a direct reading. The left inequality, $\kappa_R-\kappa_N<a+b-c$, says that
the cost disadvantage of disciplined trading must not exceed the advantage it enjoys when it is rare and noise
is abundant. The right inequality, $\kappa_N-\kappa_R<c+d-a$, says symmetrically that the cost advantage of
disciplined trading must not exceed the advantage heuristic trading enjoys when it is rare. Frictions enter
only through the difference $\kappa_R-\kappa_N$: a friction that burdens both types equally does not affect
the long-run composition.

The economic intuition is negative frequency dependence. When N-type behavior is widespread, mispricing is
abundant and disciplined strategies are highly profitable, so the R-type expands. As it grows, arbitrage
opportunities are competed away, while the liquidity it supplies makes heuristic trading less costly, so the
advantage shifts back toward the N-type. The market therefore settles at an intermediate composition, and
bounded rationality persists as a stable component of the market ecology rather than being eliminated as in
Friedman's argument.

\subsection{Antagonistic interaction: tipping threshold}

\begin{remark}[Bistability]\label{rem:1}
Suppose $b+d<0$, i.e., abundant noise trading deteriorates price signals for disciplined rules, and abundant
disciplined trading tightens arbitrage conditions against heuristic trading. If $b+d<\theta<0$, then
$x^*\in(0,1)$ still exists, but \eqref{eq:lin} is positive, so $x^*$ is unstable. Trajectories with
$x(0)>x^*$ converge to $1$ and those with $x(0)<x^*$ converge to $0$.
\end{remark}

In this case $x^*$ is not a long-run state but a tipping threshold separating two one-type regimes. Small
parameter changes that move $x^*$ past the current composition can trigger an abrupt, self-reinforcing shift
of the whole market. The sign of the cross-type coupling, rather than the mere existence of an interior fixed
point, determines whether bounded rationality persists.

\subsection{Comparative statics}

\begin{proposition}[Comparative statics]\label{prop:2}
Suppose $b+d>0$ and \eqref{eq:range} holds. Then
\begin{align}
\frac{\partial x^*}{\partial \kappa_R}&=-\frac{1}{b+d}, &
\frac{\partial x^*}{\partial \kappa_N}&=\frac{1}{b+d}, \nonumber\\
\frac{\partial x^*}{\partial b}&=\frac{1-x^*}{b+d}, &
\frac{\partial x^*}{\partial d}&=-\frac{x^*}{b+d}.
\label{eq:cs}
\end{align}
In particular, a reduction of the friction $\kappa_R$ raises the long-run R-type share, but the N-type is
eliminated only if the cost advantage of the R-type reaches $\kappa_N-\kappa_R\ge c+d-a$.
\end{proposition}

\begin{proof}
The derivatives follow from \eqref{eq:xstar}; for $b$, numerator and denominator both increase by one unit,
giving $[(b+d)-\theta]/(b+d)^2$. The last statement is the right inequality of \eqref{eq:range} together with
Proposition~\ref{prop:1}(ii).
\end{proof}

\section{Numerical Illustration}

Fig.~\ref{fig:two} illustrates the results with $a=c=1$. In the complementary case
(Fig.~\ref{fig:two}(a); $b=1$, $d=0.6$, $\kappa_R=0.4$, $\kappa_N=0$), $\theta=0.6$ and $x^*=0.375$; all
trajectories converge to $x^*$ regardless of the initial composition. In the antagonistic case
(Fig.~\ref{fig:two}(b); $b=-1$, $d=-0.6$, $\kappa_R=0$, $\kappa_N=0.4$), $\theta=-0.6$ and the same value
$x^*=0.375$ separates the basins of the two one-type regimes. Fig.~\ref{fig:two}(c) shows the long-run
composition in the complementary case as a function of the friction difference: it decreases linearly with
slope $-1/(b+d)$, and the market becomes one-type only outside the interval $-0.6<\kappa_R-\kappa_N<1$
given by \eqref{eq:range}.

\begin{figure*}[t]
\centering
\includegraphics[width=0.95\textwidth]{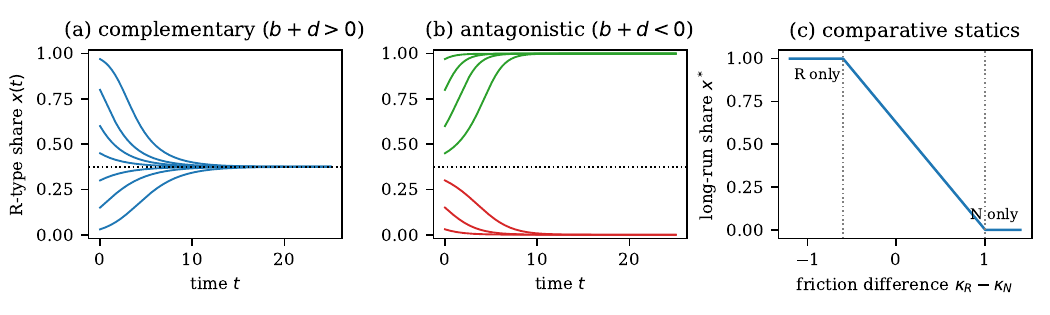}
\caption{(a) Complementary interaction: all trajectories converge to the mixed composition $x^*=0.375$
(dotted). (b) Antagonistic interaction: $x^*$ is a tipping threshold; trajectories starting above it converge
to the R-only market and those below to the N-only market. (c) Long-run R-type share in the complementary case
as a function of the friction difference $\kappa_R-\kappa_N$.}
\label{fig:two}
\end{figure*}

\section{Implications for Algorithmic and AI Trading}

Algorithmic and AI-based trading lower information-processing, monitoring, and execution costs. In the
model, this corresponds mainly to a fall in $\kappa_R$, the friction faced by disciplined strategies.
Proposition~\ref{prop:2} shows that such technological progress shifts the mixed ecology toward the R-type
continuously, but does not remove bounded rationality as long as the benefit $d$ that heuristic traders derive
from disciplined traders remains large enough. Technological progress should therefore not be read as a
monotonic movement toward universal rationality.

A more drastic effect arises if technology changes the \emph{sign} of the coupling. When many traders use
similar data and similar models, responses become faster and more synchronized, and crowding of the kind seen
in quantitative strategies in August 2007~\cite{khandani2011} may turn the interaction from complementary to
antagonistic. By Remark~\ref{rem:1}, the stable mixed ecology is then replaced by a bistable market whose
long-run composition depends on history, and small parameter changes can trigger abrupt regime shifts.

In this paper, technology enters only through the parameters of the two-type model. A natural next step is to
treat AI traders as a third behavior type that competes with human traders under the same selection mechanism,
which we leave to future work.

\section{Conclusion}

We proposed an evolutionary game-theoretic model in which bounded rationality is treated not as an anomaly
but as an endogenous component of market ecology. Under replicator dynamics, boundedly rational behavior
coexists with disciplined strategies in the long run when the cross-type interaction is complementary and the
friction difference is moderate (Proposition~\ref{prop:1}); under antagonistic interaction, the market is
bistable (Remark~\ref{rem:1}). Cost reductions brought by algorithmic and AI trading shift the composition
continuously but eliminate bounded rationality only beyond an explicit threshold (Proposition~\ref{prop:2}).

Future work includes introducing AI traders as an explicit behavior type, adding a price-formation mechanism
that links composition to returns and volatility, and confronting the model with data on the composition of
trading by investor type.

\end{document}